\documentclass[10pt,conference,romanappendices]{IEEEtran}
\usepackage{bm,cite}
\usepackage{amsmath}
\usepackage{amsthm}
\usepackage{amsbsy}
\usepackage{epsfig}
\usepackage{latexsym}
\usepackage{amssymb}
\usepackage{wasysym}
\usepackage{placeins}

\newtheorem{theorem}{Theorem}
\newtheorem{definition}{Definition}
\newtheorem{proposition}{Proposition}

\DeclareMathAlphabet{\mathbit}{OML}{cmr}{bx}{it}
\DeclareMathAlphabet{\mathsf}{OT1}{cmss}{m}{n}
\DeclareMathAlphabet{\mathTXf}{OT1}{cmss}{bx}{it}

\DeclareMathOperator*{\maximize}{maximize}
\DeclareMathOperator*{\minimize}{minimize}

\newcommand{\norm}[1]{\lVert{#1}\rVert}

\newcommand{\MG}{{\text{M}_{\text{G}}}}

\newcommand{\E}{{\text{E}}}

\newcommand{\trans}{{\text{T}}} 
 
\newcommand{\He}{{{\text{H}}}}

\begin{document} 
\title{Towards Optimal CSI Allocation in Multicell MIMO Channels}
\author{Paul de Kerret and David Gesbert\\Mobile Communications Department, Eurecom\\
 2229 route des Cr\^etes, 06560 Sophia Antipolis, France\\\{dekerret,gesbert\}@eurecom.fr}

\maketitle

\begin{abstract}
In this work\footnote{This work has been performed in the framework of the European research project ARTIST4G, which is partly funded by the European Union under its FP7 ICT Objective 1.1 - The Network of the Future.}, we consider the joint precoding across $K$ transmitters (TXs), sharing the knowledge of the user's data symbols to be transmitted towards $K$ single-antenna receivers (RXs). We consider a distributed channel state information (DCSI) configuration where each TX has its own local estimate of the overall multiuser MIMO channel. The focus of this work is on the optimization of the allocation of the CSI feedback subject to a constraint on the total sharing through the backhaul network. Building upon the Wyner model, we derive a new approach to allocate the CSI feedback while making efficient use of the pathloss structure to reduce the amount of feedback necessary. We show that the proposed CSI allocation achieves good performance with only a number of CSI bits per TX which does not scale with the number of cooperating TXs, thus making the joint transmission from a large number of TXs more practical than previously thought. Indeed, the proposed CSI allocation reduces the cooperation to a \emph{local} scale, which allows also for a reduced allocation of the user's data symbols. We further show that the approach can be extended to a more general class of channel: the exponentially decaying channels, which model accuratly the cooperation of TXs located on a one dimensional space. Finally, we verify by simulations that the proposed CSI allocation leads to very little performance losses.
\end{abstract}  
\IEEEpeerreviewmaketitle
\section{Introduction}

Network or Multicell MIMO methods, whereby multiple interfering transmitters (TXs) share user messages and allow for joint precoding, are currently considered for next generation wireless networks \cite{Gesbert2010}. With perfect message and channel state information (CSI) sharing, the different TXs can be seen as a unique virtual multiple-antenna array serving all receivers (RXs), in a multiple-antenna broadcast channel (BC) fashion. However, the allocation of the data symbols and the CSI to the cooperating TXs impose huge requirements on the architecture.%, particularly as the number of cooperating TXs increases. 
The common solution is to use disjoint clusters to reduce the amount of data to be shared\cite{Zhang2009,Papadogianni2008}. Yet, clustering limits the performance of the system because of the interference created at the edge of the clusters. In this work work, we focus on the limited sharing of the CSI feedback bits while the sharing of the data symbol is considered in the parallel work\cite{Gangula2011}. We will show later that our approach also leads to a reduction of the user's data sharing needed.

In recent works, adaptive allocation of the CSI feedback bits in a multicell cellular network has been studied\cite{Zhang2010,Ho2011,Bhagavatula2011a,Bhagavatula2011b}. However, we consider here the \emph{joint} precoding from TXs having their own \emph{local} channel estimates. This setting, introduced in \cite{Zakhour2010,dekerret2011_ISIT,dekerret2011_ISIT_arXiv} as the \emph{distributed CSI} (DCSI)-MIMO channel, opens up completely new research problems. In that case, the optimization of the CSI allocation has been discussed in \cite{dekerret2011_ISIT_arXiv} but only in terms of degrees of freedom. 
%
%In recent works, adaptive allocation of the CSI feedback bits in a multicell cellular network has been studied with coordinated beamforming\cite{Zhang2010,Ozbek2010,Ho2011,Bhagavatula2011a} and in the case of interfering broadcast channels \cite{Bhagavatula2011b}. However, we consider here the \emph{joint} precoding from TXs having their own \emph{local} channel estimates. This setting, introduced in \cite{Zakhour2010,dekerret2011_ISIT,dekerret2011_ISIT_arXiv} as the \emph{distributed CSI} (DCSI)-MIMO channel, opens up completely new research problems. In that case, the optimization of the CSI allocation has been discussed in \cite{dekerret2011_ISIT_arXiv} but only in terms of multiplexing gain. 

It is very intuitive that, in wireless networks, the precision with which a channel to a given RX should be known at a given TX depends on the distance between the TX and the RX. However, an analytical analysis has never been done and is the focus of this work. Particularly, we aim at answering the fundamental question: \emph{How should the CSI feedback be shared through the network?}

We start by analyzing a very simplified channel model, referred to in the literature as the \emph{Wyner} model, in which the TXs and the RXs are placed on a one-dimensional space (e.g. a line) and receive signals only from a few neighboring TXs. This simplistic model has the advantage of being more tractable while still offering valuable insights on more realistic channels. This model has been introduced in \cite{Wyner1994} and has been very successful since, particularly to model cooperation in wireless networks \cite{Bhagavatula2010,Levy2009,Bang2011,Shamai2011}. The extension of the work to more general channel models is provided in \cite{dekerret2012_ICC_arXiv}.

%Somekh2000
% more \emph{stingy} 
In this work, building above an asymptotic analysis taking into account the geometry of the network, we derive analytically for the Wyner model a CSI allocation allowing for a large reduction of the CSI required at the TXs at the cost of reduced performance losses. Moreover, our CSI allocation scales linearly in the number of TXs instead of the conventional quadratic scaling for full CSI sharing, while still performing well, and is thus an interesting solution to make large cooperations areas more feasible in practice. %Then, the approach is extended to more realistic channels. The proposed CSI allocation scales linearly in the number of cooperating TXs and is then particularly interesting to render cooperation of large number of TXs more practical.
\section{System Model}
We consider the distributed CSI (DCSI)-MIMO channel, in which $K$~transmitters (TXs) transmit \emph{jointly} using linear precoding to $K$~receivers (RXs) equipped with a single antenna and applying single user decoding. Each TX has the knowledge of the $K$ symbols to transmit to the $K$~RXs due to a pre-existing routing protocol for user-plane data. Besides the knowledge of data symbol, each TX is supposed to acquire, through an unspecified feedback or sharing mechanism, its \emph{own} estimate on the channel vectors to all the users, thus explaining the meaning of the term \emph{distributed}. 

The channel is represented by the channel matrix~$\mathbf{H}\in \mathbb{C}^{K\times K}$ and the transmission is described mathematically as

\begin{equation}
\begin{bmatrix}
y_1\\\vdots\\y_K
\end{bmatrix}
=
\mathbf{H}\bm{x} 
+
\bm{\eta}
=
\begin{bmatrix}
\bm{h}^{\He}_1\bm{x}\\\vdots\\
\bm{h}^{\He}_K\bm{x}
\end{bmatrix}
+
\begin{bmatrix}
\eta_1\\\vdots\\
\eta_K
\end{bmatrix} 
\label{eq:SM_1}
\end{equation}
where $y_i$ is the signal received at the $i$-th RX, $\bm{h}^{\He}_i \in \mathbb{C}^{1\times K}$ the channel to the $i$-th RX, $\bm{\eta}=[\eta_1,\ldots,\eta_K]^{\trans}\in \mathbb{C}^{K\times 1}$ the zero mean unit variance i.i.d. complex Gaussian noise ($\mathcal{CN}(0,1)$). $\bm{x}\in \mathbb{C}^{K\times 1}$ is the transmitted signal obtained from the symbol vector $\bm{s}=[s_1,\ldots,s_K]^{\trans}\in\mathbb{C}^{K\times 1}$ (i.i.d. $\mathcal{CN}(0,1)$) as
\begin{equation}
\bm{x}=\mathbf{T}\bm{s}=
\begin{bmatrix}\bm{t}_1&\ldots& \bm{t}_K\end{bmatrix}
\begin{bmatrix}s_1\\\vdots\\s_K\end{bmatrix} 
\label{eq:SM_2}
\end{equation}
where $\mathbf{T}\in \mathbb{C}^{K\times K}$ is the precoding matrix and $\bm{t}_i \in \mathbb{C}^{K\times 1}$ is the beamforming vector used to transmit $s_i$ to RX~$i$. We consider a per-stream power constraint $\norm{\bm{t}_i}^2=P$. 

Our focus is on the maximization of the sum rate averaged over the fading distribution, where the rate of user~$i$ reads as
\begin{equation}
R_i\triangleq\E_{\mathbf{H}}\left[\log_2\left(1+\frac{|\bm{h_i}^{\He}\bm{t}_i|^2}{1+\sum_{\ell\neq i}|\bm{h_{i}}^{\He}\bm{t}_\ell|^2}\right)\right].
\label{eq:SM_3}
\end{equation}
In this work, we study the high SNR regime such that it is interesting to introduce the \emph{Multiplexing Gain} (also called degree of freedom or prelog factor) at RX~$i$ as
\begin{equation}
\MG_i\triangleq \lim_{P\rightarrow\infty}\frac{R_i}{\log_2(P)}.
\label{eq:SM_3}
\end{equation}
\subsection{Distributed CSI and Distributed Precoding}

In the DCSI-MIMO channel, the $i$-th TX has its own individual estimate of its channel $\bm{h}_i^{\He}$ to RX~$i$ for all $i$, denoted by $\bm{h}^{(j)\He}_i$ and obtained from a \emph{Random Vector Quantization (RVQ)} using~$B^{(j)}_i$ bits. We focus on interference limited wireless networks working at high SNR so that we assume that Zero Forcing (ZF) precoders are used since they converge to the optimal precoder as the SNR increases. We denote by $\bm{t}_i^{(j)}$ the beamforming vector transmitting symbol $s_i$ computed at TX~$j$ and which reads as
\begin{equation}
\forall i\in\{1,\ldots,K\},\;\bm{t}_i^{(j)}\triangleq \sqrt{P}\frac{\left(\mathbf{H}^{(j)}\right)^{-1}\bm{e}_i}{\norm{\left(\mathbf{H}^{(j)}\right)^{-1}\bm{e}_i}} .
\label{eq:SM_4}
\end{equation}
Although a given TX~$j$ may compute the whole precoding matrix $\mathbf{T}^{(j)}$, only the $j$-th row will be used in practice. Finally, the effective precoder is given by
\begin{equation}
\mathbf{T}\triangleq\begin{bmatrix}\bm{t}_1&\ldots&\bm{t}_K\end{bmatrix}\triangleq\begin{bmatrix}T_{11}^{(1)}&T_{12}^{(1)}&\ldots&T_{1K}^{(1)}\\\vdots&\vdots&\ddots&\vdots\\T_{K1}^{(K)}&T_{K2}^{(K)}&\ldots&T_{KK}^{(K)}\end{bmatrix}.
\label{eq:SM_6}
\end{equation} 
More detail on precoding for the DCSI-MIMO channel can be found in \cite{dekerret2011_ISIT_arXiv}.

\subsection{Optimized CSI Allocation}
Depending on the CSI accuracy at a TX, the coefficients implemented will, or will not, be very different from the coefficients obtained with perfect CSI. Considering a realistic constraint on the total number of CSI bits transmitted via the multiuser feedback channel, which we denote by $B_{\max}$, we want to find the allocation of the feedback bits maximizing the average sum rate. Thus, the optimization problem is then:
\begin{equation}
\maximize_{\{B_i^{(j)}\}}  \sum_{i=1}^K R_i \text{, s.t. $\sum_{j=1}^K\sum_{i=1}^K B_i^{(j)}\leq B_{\max}$}.
\label{eq:SM_8}
\end{equation}
However, this problem is intricate as the dependency of the sum rate as a function of the CSI estimate error does not have an easy structure. Based on the insight from the asymptotic analysis in the DCSI-MIMO channel~\cite{dekerret2011_ISIT_arXiv}, we approximate this optimization problem by the following one:
\begin{equation}
\minimize_{\{B_i^{(j)}\}} \sum_{i=1}^K \E[\norm{\bm{t}_i-\bm{t}_i^{\text{PCSI}}}^2]\text{, s.t. $\sum_{j=1}^K\sum_{i=1}^K B_i^{(j)}\leq B_{\max}$}
\label{eq:SM_9}
\end{equation}
where $\bm{t}_i^{\text{PCSI}}$ is the beamforming vector transmitting symbol $s_i$ using perfect CSI and is given by
\begin{equation}
\forall i\in\{1,\ldots,K\},\;\bm{t}_i^{\text{PCSI}}\triangleq \frac{\sqrt{P}}{\norm{\mathbf{H}^{-1}\bm{e}_i}}\mathbf{H}^{-1}\bm{e}_i .
\label{eq:SM_10}
\end{equation}
The intuition behind this figure of merit follows from the fact that, in the DCSI-MIMO channel, the consistency between the precoders implemented at the TXs has a critical impact on the performance.

It is widely believed that the accuracy of the CSI at a TX should somehow be a function of how much interference it actually contributes to each RX. However, despite the intuition, no previous paper has identified how the CSI quantization bits should be allocated throughout the network. Building upon some simplified models, we will show in the following that it is possible to obtain some quantitative results which can then be used as guidelines for practical systems.

\section{Wyner Model} \label{se:Wyner}
We consider a finite linear version of the Wyner model where $K$~TXs equipped with only one antenna are uniformly distributed along a line. A RX receives interference coming from its two direct neighboring TXs with an inter-cell attenuation factor equal to $\mu\in (0,1)$, while the short term fading is assumed to be Rayleigh distributed. Thus, the multiuser channel matrix $\mathbf{H}\in\mathbb{C}^{K\times K}$ is a tridiagonal matrix defined as
\begin{equation*}
\mathbf{H}\triangleq
\begin{bmatrix}
d_1&a_1\mu&0&\ldots&\ldots&\ldots&0\\
\mu b_2&d_2&\mu a_2&0&\ldots&\ldots&0\\ 
\ldots&\ddots&\ddots&\ddots&0&\ldots&0\\
\ldots&0&\mu b_i&d_i&\mu a_i&0&\ldots\\
0&\ldots&0&\ddots&\ddots&\ddots&0\\
0&\ldots&\ldots&0&\mu b_{K-1}&d_{K-1}&\mu a_{K-1}\\
0&\ldots&\ldots&\ldots&0&\mu b_K&d_K\\
\end{bmatrix}
\label{eq:Wyner}
\end{equation*}
where $b_i$, $d_i$, and $a_i$ are distributed as standard i.i.d. complex Gaussian ($\mathcal{CN}(0,1)$). %Note that each channel vector contains only $3$~non-zeros elements (except for the extremal TXs which have only two but we will not discuss in detail the boundary effects). 
We denote the estimates of $b_i$, $d_i$, and $a_i$ at TX~$j$ by $b_i^{(j)}$, $d_i^{(j)}$, and $a_i^{(j)}$, respectively, while the estimate errors are then $\Delta{b_i^{(j)}}$, $\Delta{d_i^{(j)}}$, and $\Delta{a_i^{(j)}}$, respectively.

\subsection{Inverse of a Tridiagonal Matrix}\label{se:Wyner:Inv}
The Wyner model allows to obtain a closed form for the channel inverse \cite{El-Mikkawy2005}, which we recall in the following and will be useful to quantify the effect of the distributed CSI.

We start by introducing two $K+1$ dimensional vectors $\bm{\beta}=[\beta_1,\ldots,\beta_{K+1}]$ and $\bm{\alpha}=[\alpha_0,\ldots,\alpha_K]$ which will be used in the matrix inverse and are defined as
\begin{equation}
\alpha_i=\begin{cases}
1&\text{, if i=0}\\
d_1&\text{, if i=1}\\
d_i\alpha_{i-1}-\mu^2 b_i a_{i-1}\alpha_{i-2}&\text{, if i=2,\ldots,K}\\
\end{cases}
\label{eq:Wyner_1}
\end{equation}
and
\begin{equation}
\beta_i=\begin{cases}
1&\text{, if i=K+1}\\
d_K&\text{, if i=K}\\
d_i\beta_{i+1}-\mu^2 b_{i+1} a_{i}\alpha_{i+2}&\text{, if i=1,\ldots,K-1}
\end{cases}.
\label{eq:Wyner_2}
\end{equation}

The channel inverse $\mathbf{H}^{-1}$ then reads as follows. First, the diagonal elements for $i=2,\ldots,n-1$ are 
\begin{equation}
\{\mathbf{H}^{-1}\}_{ii}\!=\!\left(d_i\!-\!\frac{\mu^2b_{i+1}a_{i}\beta_{i+2}}{\beta_{i+1}}-\frac{\mu^2b_i a_{i-1}\alpha_{i-2}}{\alpha_{i-1}}\right)^{-1}
\label{eq:CSI_2}
\end{equation}
while the two extremal diagonal elements are 
\begin{equation}
\begin{aligned}
\{\mathbf{H}^{-1}\}_{11}&=\left(d_1-\frac{\mu^2b_2a_1\beta_3}{\beta_2}\right)^{-1}\\
\{\mathbf{H}^{-1}\}_{nn}&=\left(d_n-\frac{\mu^2b_n a_{n-1}\alpha_{n-2}}{\alpha_{n-1}}\right)^{-1}.
\end{aligned}
\label{eq:Wyner_3}
\end{equation}
Last, the off-diagonal elements are given by
\begin{equation}
\begin{aligned} 
\{\mathbf{H}^{-1}\}_{ij}&=
\begin{cases}
(-\mu)^{j-i}(\prod_{k=1}^{j-i}a_{j-k})\frac{\alpha_{i-1}}{\alpha_{j-1}}\{\mathbf{H}^{-1}\}_{jj}&\text{if $i<j$}\\
(-\mu)^{i-j}(\prod_{k=1}^{i-j}b_{j+k})\frac{\beta_{i+1}}{\beta_{j+1}}\{\mathbf{H}^{-1}\}_{jj}&\text{if $i>j$}.
\end{cases}
\end{aligned}
\label{eq:Wyner_4}
\end{equation}

\subsection{CSI Allocation in the Wyner Model}\label{se:Wyner:CSI}
%With perfect CSI, the optimal solution for \eqref{eq:SM_9} would be to use the expression for the inverse given in the previous section and then normalize them according to \eqref{eq:SM_4}.
The optimization of the CSI allocation in the DCSI-MIMO channel is still involved, even in the Wyner model, as it requires evaluating the impact of the CSI estimate errors in the coefficients computed at \emph{each} TX based on its \emph{own local} CSI. Thus, we consider instead an approximation and study the CSI allocation at high SNR. We start by recalling the following asymptotic result in the DCSI-MIMO channel.
%(which generalizes the very well known result on the MG in a multiple-antenna Broadcast Channel from Jindal~\cite{Jindal2006}).
\begin{theorem} \cite{dekerret2011_ISIT_arXiv} 
To achieve the maximal MG in the DCSI-MIMO Rayleigh fading channel from $K$~TXs to $K$~RXs, it is necessary and sufficient to fulfill
\begin{equation}
\forall i\in\{1,\ldots,K\}, \lim_{P\rightarrow \infty}\E[\norm{\bm{t}_i-\bm{t}_i^{\text{PCSI}}}^2]=O(1)
\label{eq:Wyner_5}
\end{equation}
which is achieved with ZF using RVQ if and only if 
\begin{equation}
\forall i,j\in\{1,\ldots,K\}, \lim_{P\rightarrow \infty}\frac{B_{i}^{(j),\text{Full}}}{(K-1)\log_2(P)}=1.
\label{eq:Wyner_6}
\end{equation}
\label{theorem_MG}
\end{theorem}
Note that the $K-1$ in \eqref{eq:Wyner_6} is replaced here by a~$2$ as it corresponds to the non-zeros elements in the channel vector.

The optimal CSI allocation can be seen to converge to the uniform CSI allocation as the SNR increases. However, a uniform allocation is clearly very suboptimal at finite SNR, particularly in a large network with pathloss, where a distant TX emits very little interference. Therefore, the geometry of the network has to be used in order to derive an efficient CSI allocation at finite SNR.

We can observe in the expressions given in Subsection~\ref{se:Wyner:Inv} for the inverse of a tridiagonal channel that the amplitude of the coefficients in the inverse decreases exponentially as the elements get away from the diagonal. This is actually a consequence of the Wyner model which models the pathloss attenuation by considering interference only from direct neighbors. In the following, we make use of this property to derive an efficient CSI allocation. 

To derive analytical results at high SNR, we consider that the inter-cell coefficient $\mu$ can be written as a fixed fraction of $P$, i.e., $P^{-\varepsilon}$ for a certain $\varepsilon\in (0,1)$. Thus, the inter-cell coefficient decreases as the power increases. Note that this is only an artefact to model the impact of the inter-cell attenuation in the high SNR regime and that we show by simulations that the performance remain good as $\mu$ increases.% as the inter-cell attenuation factor has otherwise no impact on the performance at high SNR.

\begin{theorem} %
Considering that $\mu=P^{-\varepsilon}$ for a fixed $\varepsilon\in (0,1)$ and that random vector quantization (RVQ) is used, it is sufficient in order to achieve the maximal MG to quantize the channel vectors $\bm{h}_i^{\He}$ at TX~$j$ with a number of bits
\begin{equation}
B_i^{(j)}=\max(2\log_2(P\mu^{2|i-j|}),0).
\label{eq:Wyner_7}
\end{equation}
\label{theorem_CSI_alloc_Wyner}
\end{theorem}

\begin{proof}
We start by inserting the CSI error estimates in the off-diagonal elements of the inverse given in Subsection~\ref{se:Wyner:Inv}. We consider the coefficient with $j<i$ as the case $j>i$ follows by symmetry. In a first step we consider the coefficient at TX~$j$ corresponding to the transmission of stream~$i$.
\begin{equation}
\begin{aligned} 
&\frac{\{\mathbf{H}^{(j)-1}\}_{ji}}{\{\mathbf{H}^{(j)-1}\}_{ii}}= 
(-\mu)^{i-j}\tfrac{\hat{\alpha}_{j-1}^{(j)}+\Delta{{\alpha}_{j-1}^{(j)}}}{\hat{\alpha}_{i-1}^{(j)}+\Delta\alpha_{i-1}^{(j)}}\prod_{k=1}^{i-j}(\hat{a}^{(j)}_{i-k}+\Delta a^{(i)}_{i-k})\\
&\approx (-\mu)^{j-i}\tfrac{\hat{\alpha}_{j-1}^{(j)}+\Delta\alpha_{j-1}^{(j)}}{\hat{\alpha}_{i-1}^{(j)}+\Delta\alpha_{i-1}^{(j)}}\left(\prod_{k=1}^{i-j}\hat{a}^{(j)}_{i-k}\right)\left(\sum_{\ell=1}^{i-j}\frac{\Delta a^{(j)}_{j+\ell}}{\hat{a}^{(j)}_{j+\ell}}\right)
\end{aligned}
\label{eq:Wyner_8}
\end{equation}
where we have considered that the estimation errors are small as the CSI allocation will be derived so as to fulfill \eqref{eq:Wyner_5}, so that we can do a first order approximation in the estimation error terms. To proceed, we approximate the first ratio in \eqref{eq:Wyner_8} by keeping only the first order coefficients in $\mu$ which gives using \eqref{eq:Wyner_1}:
\begin{equation}
\begin{aligned}
\frac{\hat{\alpha}_{j-1}^{(j)}+\Delta\alpha_{j-1}^{(j)}}{\hat{\alpha}_{i-1}^{(j)}+\Delta\alpha_{i-1}^{(j)}}&\approx \prod_{k=1}^{i-j}(\hat{d}^{(j)}_{i-k}+\Delta d^{(j)}_{i-k})^{-1}\\
&\approx\left(\prod_{k=1}^{i-j}\hat{d}^{(j)}_{i-k}\right)^{-1}\left(\sum_{\ell=1}^{i-j}\frac{\Delta d^{(j)}_{j+\ell}}{\hat{d}^{(j)}_{j+\ell}}\right)^{-1}.
\end{aligned}
\label{eq:Wyner_9}
\end{equation}
To obtain an approximate closed form for the channel inverse, we also need to approximate the diagonal term defined in \eqref{eq:CSI_2} in the denominator of the Left Hand Side (LHS) in \eqref{eq:Wyner_8} which is done as 
\begin{equation}
\{\mathbf{H}^{-1}\}_{ii}\!=\!\left(d_i\!-\!\frac{\mu^2b_{i+1}a_{i}\beta_{i+2}}{\beta_{i+1}}-\frac{\mu^2b_i a_{i-1}\alpha_{i-2}}{\alpha_{i-1}}\right)^{-1}\approx\frac{1}{d_i}.
\label{eq:Wyner_10}
\end{equation}

To fulfill condition \eqref{eq:Wyner_5}, the sum of all the error terms arising in \eqref{eq:Wyner_8} has to tend to zero as $O(1/P)$ due to the fact that the channel inverse elements given in \eqref{eq:Wyner_8} are further normalized so as to fulfilled the \emph{per-stream} power constraint of~$P$. 

Inserting \eqref{eq:Wyner_9} and \eqref{eq:Wyner_10} in \eqref{eq:Wyner_8} we conclude by inspection that this condition is fulfilled on the first order when $a_{\ell}^{(j)}$ and $d_{\ell}^{(j)}$ for $\ell\in \{i-1,i-2,\ldots,j\}$ as well as $d^{(j)}_i$ are obtained from the quantization of the channel vector $\bm{h}_{\ell}$  for $\ell\in \{i-1,i-2,\ldots,j\}$ with $\max(2\log_2(P\mu^{2(i-j)}),0)$ bits. 

This conditions corresponds to the transmission of symbol $s_i$ at TX~$j$ and is a consequence that the average power used for that coefficient is $O(P\mu^{2(i-j)})$ instead of $O(P)$ such that the accuracy of the CSI can be reduced.

These requirements on the CSI quantization correspond to the strongest requirement for that stream at that TX, but the other channel parameters need obviously to be also known, although with a lower accuracy. Particularly, no CSI requirement has been made on the $\hat{b}^ {(j)}_{\ell}. $It is thus necessary to improve the accuracy of the approximation done in \eqref{eq:Wyner_8} and \eqref{eq:Wyner_9}. Using the definition of $\alpha_{i}$, we can write
\begin{equation}
\begin{aligned}
\frac{\hat{\alpha_{j}}^{(j)}}{\hat{\alpha_{i}}^{(j)}}&=\frac{\hat{\alpha_{j}}^{(j)}}{\hat{d}_i^{(j)}\hat{\alpha}^{(j)}_{i-1}-\mu^2 \hat{b}^{(j)}_i \hat{a}^{(j)}_{i-1}\hat{\alpha}^{(j)}_{i-2}},\\
&\approx\frac{1}{\hat{d}_i^{(j)}\hat{\alpha}^{(j)}_{i-1}}\left(1+\frac{\mu^2 \hat{b}^{(j)}_i \hat{a}^{(j)}_{i-1}\hat{\alpha}^{(j)}_{i-2}}{\hat{d}_i^{(j)}\hat{\alpha}^{(j)}_{i-1}}\right)\\
\end{aligned}
\label{eq:Wyner_10}
\end{equation}
after keeping this time the first order term. The ratio of the two $\hat{\alpha}^{(j)}_{\ell}$ can then be approximated by the product of the $\hat{d}_{\ell}$ using the zero-th order approximation as in \eqref{eq:Wyner_9}. Inserting this expression in \eqref{eq:Wyner_8}, we obtain the additional conditions that $\hat{b}_i^{(j)}$, $\hat{a}_{i-1}^{(j)}$, $\hat{d}_i^{(j)}$ and $\hat{d}_{i+1}^{(j)}$ should be resulting from the quantization of the channel vector with at least~$\max(2\log_2(P\mu^{2(i-j+2)}),0)$ bits, from which only the conditions on $\hat{b}_i^{(j)}$ and $\hat{d}_{i+1}^{(j)}$ are new. 

Proceeding similarly for increasing number of terms kept, we increase the accuracy of the approximation made in \eqref{eq:Wyner_10} to the first for order to obtain the condition that $\hat{b}_{i+1}^{(j)}$, $\hat{a}_{i}^{(j)}$, $\hat{b}_{i}^{(j)}$, and $\hat{a}_{i-1}^{(j)}$ should be quantized from a RVQ with at least~$\max(2\log_2(P\mu^{2(i-j+2)}),0)$ bits, from which only the conditions on $\hat{b}_{i+1}^{(j)}$ and $\hat{a}_{i}^{(j)}$ are new.

Increasing further the accuracy of the approximation leads to a requirement for each term. However, for a given TX, this has to be done for all the streams to be transmitted. Thus, only the strongest requirements are of real interest in our case. Considering only the three constraints derived previously to have $\hat{a}^{j}_{i-1}$ and $\hat{d}^{j}_i$ from a RVQ with $\max(2\log_2(P\mu^{2(i-j)}),0)$ bits, and $\hat{b}^{j}_{i+1}$ from a RVQ with $\max(2\log_2(P\mu^{2(i-j+2)}),0)$ bits, we can see that considering only these three constraint for all the streams transmitted leads to the most constraining requirements. 

Additionnaly, as the elements from the vector are quantized jointly, we obtain the final requirement for the channel inverse. The last point to discuss comes from the normalization following the power constraint fullfilment. Indeed, each vector of the channel inverse if divided by its norm and multiplied by $\sqrt{P}$. On the zero-th order, the norm of the channel inverse beamformer is equal to $1/d_i$, which has then to be quantized with $\max(2\log_2(P\mu^{2(i-j)}),0)$ bits. This is already the case which conclude the proof. 
\end{proof}

\emph{Remark $1$:} In fact, the difference of required accuracy between $d_j^{(j)}$, and the off-diagonal elements $a_j^{(j)}$ and $b_j^{(j)}$ comes only from the fact that the attenuation factor $\mu$ is known such that the outer-diagonal elements can be factorized into the attenuation factor $\mu$ and the Rayleigh fading part. Thus, the estimation error is directly attenuated by the inter-cell attenuation factor $\mu$. 

yet, this result should be interpreted in the right way which is that the estimation error for the three coefficients $\mu a_j^{(j)}$, $d_j^{(j)}$ and $\mu b_j^{(j)}$ should have the same variance. Another way to understand it, is that the variance of the interfering coefficient is smaller so that the number of bits necessary to quantize it with the \emph{same accuracy} is smaller. Using random scalar quantization for each coefficient with the knowledge of the attenuation factor, the condition given in the proposition can be replaced by the condition
\begin{equation}
B_i^{(j)}=\sum_{k=1}^3\max(\log_2(P\mu^{2|k-j|}\mu^{2|i-j|}),0).
\end{equation}

\emph{Remark $2$:} A \emph{per stream} power constraint has been considered instead of a constraint for the constraint over the whole precoder. An additional \emph{per antenna} constraint could be further introduced but this would complicate the analysis. The property that has to be fulfilled to allow for the optimization with only local CSI is that the optimization is considered for each stream separately. Indeed, using a total power constraint would lead the TXs to normalize the precoder using the beamforming vectors corresponding to users very far away, i.e., using very inaccurate CSI. This is problematic in the DCSI-MIMO channel as this reduces the consistency between the TXs.

Intuitively, the CSI allocation is based on the fundamental property that a given TX contributes a large power only to the streams transmitted to neighboring RXs. For the symbols transmitted to distant RXs, it contributes to the transmission of the stream with only a small power, so that a less accurate CSI is needed to maintain a constant interference level created by the transmission of that stream.

The CSI allocation derived in Theorem~\ref{theorem_CSI_alloc_Wyner} reduces significantly the requirements on the CSI. Particularly, we will show in the following that it leads to the transmission of the CSI only at a local scale, thus giving a CSI by TX whish does not scale with the number of TXs. In relation with this fundamental property, we call the CSI allocation given in Theorem~\ref{theorem_CSI_alloc_Wyner} as the \emph{Decaying CSI allocation}.

\subsection{CSI Scaling in the Number of Cooperating TXs} 
In the following, we consider the total number of feedback bits shared through the backhaul network which is of interest to us, as it is a benchmark of the feasibility of the cooperation. Particularly, we are interested in the scaling behavior of the total number of bits for large number of cooperating TXs. In the full CSI allocation given in Theorem~\ref{theorem_MG}, it can be shown to be equal to:
\begin{equation}
B^{\text{Full}}=\sum_{j=1}^K B^{(j),\text{Full}}=K^2\log_2(P)
\label{eq:CSI_scale_1}
\end{equation}
where we have denoted by $B^{(j)}$ the total number of feedback bits allocated to TX~$j$. The total number of feedback bits $B^{\text{Full}}$ increases quadratically with the number of cooperating TXs. This is problematic for large cooperation areas and is no longer the case in the proposed CSI allocation.%, more \emph{parsimonious}, 

\begin{proposition}
For finite $P$ and $\mu$, the CSI allocation given in Theorem~\ref{theorem_CSI_alloc_Wyner} scales linearly in the number~$K$ of cooperating TXs, which means that the number of CSI bits required by TX/RX pairs does not increase with the number of cooperating TX~$K$.  
\label{proposition_scaling_wyner}
\end{proposition}
\begin{proof}
%A detailled proof is given in \cite{dekerret2011_WCNC_arXiv}.
%A detailled proof is given in Appendix~\ref{se:App:proposition_scaling_wyner}. 
Let consider $\mu$ and $P$ to be given. The attenuation factor $\mu$ takes its value in~$(0,1)$, so that it exists a natural number $K_0$ verifying $\mu^{K_0}<1/\sqrt{P}$. We consider $K$ to be very large (and odd to simplify the notation) such that we can consider wlog a TX placed in the middle of the cooperation domain. Using the CSI allocation from Theorem~\ref{theorem_CSI_alloc_Wyner}.
\begin{equation}
\begin{aligned}
B^{(j)}=\sum_{i=1}^K B_i^{(j)}&=\sum_{i=1}^K\max(2\log_2(P\mu^{2|i-j|}),0)\\
&=2\sum_{i'=1}^{K/2}\max(2\log_2(P\mu^{2i'}),0).
\end{aligned}
\end{equation}
Letting $K$ tend to infinity, the sum has only a finite number of positive terms as the logarithm will be negative for $i'>K_0$. Thus the $B^{(j)}$ do not scale with $K$ which ends the proof.
\end{proof}

The insight behind the limited CSI allocation by TX is that the exponential decay of the off-diagonal coefficients in the precoding matrix is exploited to reduce the cooperation to a \emph{local scale}. 

\subsection{User's Data Symbol Allocation}
In this work, we have considered the user's data symbols to be shared between all the TXs. This is very constraining particularly when considering networks of large size.

yet, our analysis for the CSI allocation has introduced a notion of \emph{local Cooperation} which will use in the following to derive a routing solution for the user's data symbol which also requires only local Cooperation.

\begin{proposition}
Considering that $\mu=P^{-\varepsilon}$ for a fixed $\varepsilon\in (0,1)$, sharing the symbol $s_i$ only to the TXs which verify $2\log_2(P\mu^{2|i-j|})>0$ leads to no loss in MG. 

Furthermore, for finite $P$ and $\mu$ and increasing $K$, this user's data symbol sharing scales linearly in the number of TX/RX pairs, i.e., each symbol is shared only to a finite number of cooperating TXs. 
\label{proposition_data_wyner}
\end{proposition}
\begin{proof}
The proof follows easily from the proof of Theorem~\ref{theorem_CSI_alloc_Wyner} because when $2\log_2(P\mu^{2|i-j|})<0$, TX~$j$ is allocated with no CSI  such that this would lead to a MG if the power used to transmit that symbol were positive. Thus, TX~$j$ does not participate to the transmission of symbol~$s_i$, which therefore doesn't need to be shared to TX~$j$.

The finite number of cooperating TXs follows then directly from the same property for the CSI allocation given in Propositions~\ref{proposition_scaling_wyner}.
\end{proof}

The property of sharing the symbols only to a finite number of TXs is interesting as combined with the local CSI allocation, it gives a complete alternative to clustering algorithms. Both solutions have similar requirements for the backhaul and can be fairly compared.

Thus, our alternative clustering solution based on the \emph{Decaying CSI allocation} could then be implemented in realistic settings. In fact, we will see in the following sections, that the results can be generalized to more realistic settings.

\subsection{Extension to More General Band Limited Channel Models}
The analysis was carried out in a symmetric and regular linear Wyner model but can be extended without difficulties to aribtrary banded matrices using bounds on the amplitude of the off-diagonal coefficients\cite{Demko1984}, which we recall in the following.
\begin{theorem}[Demko$1984$]
Let $\mathbf{A}$ be an operator on a linear space, $m$-banded, centered, bounded and boundedly invertible, then
\begin{equation}
\{\mathbf{A}^{-1}\}_{i,j}\leq C \lambda^{|i-j|}
\end{equation}
where 
\begin{equation}
\begin{aligned}
\lambda&=\left(\frac{\text{cond}(\mathbf{A})-1}{\text{cond}(\mathbf{A})+1}\right)^{\frac{1}{m}}\\
C&=(m+1)\lambda^{-m}\|\mathbf{A}^{-1}\|\text{cond}(\mathbf{A})\max\left(1,\left[\frac{1+\text{cond}(\mathbf{A})}{\sqrt{2}\text{cond}(\mathbf{A})}\right]^2\right).
\end{aligned}
\end{equation}
\label{theorem_Demko}
\end{theorem}
This bound can be used to determine a reduced CSI allocation similarly to previous done. Due to the dependency in the condition number, the bound has to be recalculated for each channel realization but this would also be the case in the previous analysis if no asumption was being made on the value of the attenuation factor $\mu$.

The exponential decay of the amplitude of the channel inverse is proven in Theorem~\ref{theorem_Demko} for arbitrary matrices. This property is the basis for our approach such that the decaying CSI allocation can be extended to banded matrices. For brievity concerns, this is not further discussed.

\subsection{Simulations for the Wyner Model}\label{se:Wyner:Sim}
We simulate the average sum rate in the Wyner model presented at the beginning of Section~\ref{se:Wyner} for a network made of $25$ cooperating TXs applying distributively a ZF precoder as given in \eqref{eq:SM_4}. For comparison, we show the performance achieved when \emph{perfect CSI} is available at all TXs. To evaluate the efficiency of the \emph{Decaying CSI allocation} proposed in Theorem~\ref{theorem_CSI_alloc_Wyner}, we also compare it to two alternative CSI allocations using the same total number of feedback bits: the \emph{uniform CSI allocation} and the \emph{overlapping cluster} allocation. Clearly, the first one consists in allocating the bits uniformly. The second one consists in an improved clustering solution where TX~$j$ has the knowledge of the channel vector $\bm{h}_i^{\He}$ for $i\in\{j-n_{\text{cluster}},\ldots,j+n_{\text{cluster}}\}$ with the number of bits uniformly allocated to all TXs and channels. In the simulations, we consider $n_{\text{cluster}}=2$, such that the number of cooperating $TX$ is $5$. This allocation has been used to improve on the usual clustering and is itself already a novel CSI allocation. Note that this allocation will be seen to perform quite badly in the Wyner model, but to have very good performance in the more realistic model discussed in Section~\ref{se:ExpDec}
\begin{figure}%[htp!] 
\centering
\includegraphics[width=1\columnwidth]{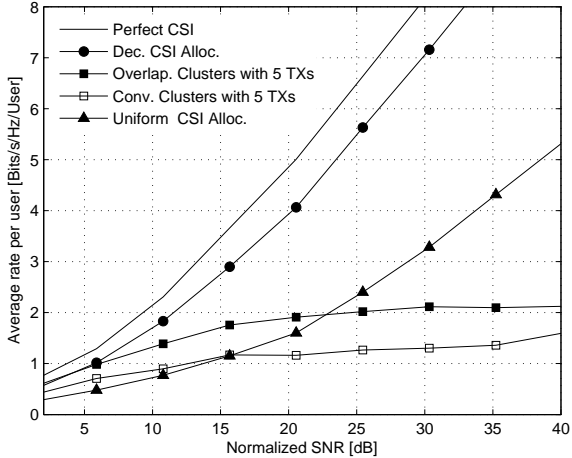}
\caption{Average rate per user in terms of the normalized transmit power per TX for an inter-cell attenuation factor $\mu=0.5$.}
\label{Rate_Wyner_25TX_musqrt05}
\end{figure}

In Fig.~\ref{R_Wyner_25TX_Musqrt05_VarPow}, we plot the average rate per user as a function of the normalized transmit power for $\mu=0.5$. The \emph{Decaying CSI allocation} outperforms clearly the \emph{overlapping cluster} solution. In fact, the decaying allocation could be shown to perform better than the overlapping cluster allocations for all the size of the clusters and is somehow approximated by an overlapping cluster allocation with the size changing in terms of the SNR.

The uniform CSI allocation becomes efficient only at very high SNR, when the pathloss does not any longer have a strong impact compared to the SNR, while the decaying CSI allocation performs well at all SNRs. The percentage of CSI used in the decaying CSI allocation compared to the full CSI allocation in \eqref{eq:Wyner_5} increases as the SNR increases. Indeed, the intercell attenuation factor $\mu$ is fixed, such that the CSI allocation converges to the uniform CSI allocation at high SNR. For a transmit power per TX of $20$ dB with $\mu=0.5$ it is equal to $0.43$ and to $0.63$ for a transmit power of $45$ dB. Furthermore, the percentage of CSI allocated also clearly decreases when the value of $\mu$ decreases such that the solution proposes becomes even more interesting.

\begin{figure}%[htp!] 
\centering
\includegraphics[width=1\columnwidth]{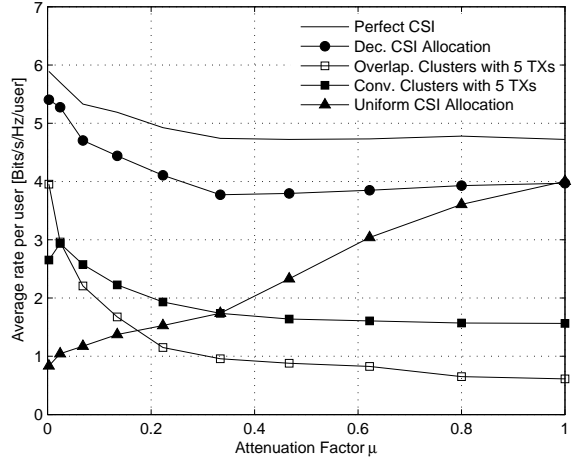}
\caption{Average rate per user in terms of the inter-cell attenuation factor $\mu$ for a normalized transmit power per TX equal to $20$ dB.}
\label{R_Wyner_25TX_P30dB_VarMu}
\end{figure}
In Fig.~\ref{R_Wyner_25TX_P30dB_VarMu}, we have plotted the average rate per user as a function of the inter-cell attenuation factor $\mu$ for a given normalized transmit power per TX equal to $20$~dB. As the attenuation factor $\mu$ becomes larger, the difference between the perfect CSI configuration and the one with optimized CSI allocation increases but for all values of $\mu$ the optimized CSI allocation remains very performing compared to the other solutions. As expected, when $\mu=1$, there is no inter-cell attenuation and the optimized CSI allocation becomes equal to the uniform CSI allocation.

%%%%%%%%%%%%%%%%%%%%%%%%%%%%%%%%%
%%%%%%%%%%%%%%%%%%%%%%%%%%%%%%%%%
%%%%%%%%%%%%%%%%%%%%%%%%%%%%%%%%%
%%%%%%%%%%%%%%%%%%%%%%%%%%%%%%%%%
%%%%%%%%%%%%%%%%%%%%%%%%%%%%%%%%%
%%%%%%%%%%%%%%%%%%%%%%%%%%%%%%%%%

\section{CSI Allocation in Exponentially Decaying Channels}\label{se:ExpDec}

%%%%%%%%%%%%%%%%%%%%%%%%%%%%%%%%%
%%%%%%%%%%%%%%%%%%%%%%%%%%%%%%%%%
%%%%%%%%%%%%%%%%%%%%%%%%%%%%%%%%%
%%%%%%%%%%%%%%%%%%%%%%%%%%%%%%%%%
%%%%%%%%%%%%%%%%%%%%%%%%%%%%%%%%%
%%%%%%%%%%%%%%%%%%%%%%%%%%%%%%%%%
%%%%%%%%%%%%%%%%%%%%%%%%%%%%%%%%%
%%%%%%%%%%%%%%%%%%%%%%%%%%%%%%%%%

The Wyner model represents an interesting model for its simplicity but is clearly a simplistic modeling of the real transmission. One of the main limiting asumption is the limited local interaction between the TXs and we will now introduce a model which lifts this restriction. We consider thus a more realistic channel model where all the links are non-zeros but the channel matrix has an \emph{exponentially decaying} structure. This means that the amplitude of the $(i,j)$-th element of the channel matrix decreases in an exponentially fashion as a function of the difference $|i-j|$. Intuitively, this means that the weight of a coefficient descreases as it goes away from the diagonal.

In this section, we will consider for simplicity a regular model for the exponential decaying matrix. The channel matrix $\mathbf{H}$ is then made from the element-wise product of a standard Rayleigh fading channel matrix $\mathbf{G}\in\mathbb{C}^{K\times K}$ and a Toeplitz matrix generated from the vector~$[1,\mu,\mu^2,\ldots,\mu^{K-1}]$:

{\small\begin{equation}
\mathbf{H}\triangleq
\begin{bmatrix}
G_{11}&\mu G_{12}&\mu^2 G_{13}&\ldots&\mu^{K-1} G_{1K}\\
\mu G_{21}& G_{22}&\mu G_{23}&\ldots&\mu^{K-2} G_{2K}\\ 
\mu^3 G_{31}& \mu G_{32}& G_{33}&\ldots&\mu^{K-3} G_{3K}\\ 
\vdots&\vdots&\vdots&\ddots&\vdots\\
\mu^{K-1} G_{K1}&\mu^{K-2} G_{K2}&\mu^{K-3}  G_{K3}&\ldots&G_{KK}
\end{bmatrix}.
\label{eq:ExpDec_1}
\end{equation}}	

Such a channel matrix corresponds to practical transmission scenarios, e.g., from the joint transmission from a set of TXs set regularly along a road. In a first part, we will start by giving some mathematical properties of the exponetially decaying and we will then show how these results can be used to extend the approach from the previous Section to that setting.

\subsection{Mathematical Prerequisite on Exponentially Decaying Matrices}
We start by defining rigorously an exponentially decaying matrix.
\begin{definition}
An infinitely large matrix $\mathbf{A}$ is said to be \emph{exponentially decaying} if it exists a $\gamma>0$ such that
\begin{equation}
\forall \gamma'<\gamma,|\{\mathbf{A}\}_{i,j}|\leq c'(\gamma')exp(-\gamma'|i-j|)
\label{eq:ExpDec_2}
\end{equation}
\end{definition}
where $c'(\gamma')$ is some constant proper to the matrix $\mathbf{A}$. In \cite{Groechenig2006,Groechenig2010,Bickel2010}, an involved functional analysis is done to extend the particularly interesting work \cite{Jaffard1990} from which the relevant properties are recalled here.
\begin{proposition}\cite{Jaffard1990}
Let $\mathbf{A}$ be an invertible matrix exponentially decaying with a coefficient $\gamma$, then $\mathbf{A}^{-1}$ is exponentially decaying for a coefficient $\gamma'$.
\label{proposition_DecayRate}
\end{proposition} 
This proposition allows to obtain information on the structure of the channel inverse, which will be our ZF precoder, based on the structure of the channel. It implies that the set of exponentially decaying matrices is closed under inversion. Note however, that the constant $\gamma'$ is a priori different from $\gamma$, but we will in the following assume the constrant $\gamma'$ to be known. As this constant depends only on the pathloss attenuation, it appears a realistic assumption that a long term channel estimate is inverted so as to estimate the value of the decay rate for the inverse $\gamma'$.

\begin{proposition}\cite{Jaffard1990}
Let $\mathbf{A}$ and $\mathbf{B}$ be two exponentially decaying matrices for a coefficient $\gamma$ and let $\Omega$ be any set of indices. If
\begin{equation}
|\{\mathbf{A}\}_{i,j}-\{\mathbf{B}\}_{i,j}|\leq c\exp(-\gamma(d(i,\Omega)+d(j,\Omega)))
\label{eq:ExpDec_3}
\end{equation}
where $c$ is some arbitrary constant and the distance $d(\cdot,\Omega)$ denotes the minimal distance to an element in the set $\Omega$. It holds then
\begin{equation}
|\{\mathbf{A}^{-1}\}_{i,j}\!-\!\{\mathbf{B}^{-1}\}_{i,j}|\leq c'\exp(-\gamma'(d(i,\Omega)+d(j,\Omega)))
\label{eq:ExpDec_4}
\end{equation}
\label{proposition_window}
\end{proposition}
where $c'$ is another constant. This proposition which was called in \cite{Jaffard1990} the \emph{Windows Lemma} can be in fact very intuitively (and linked to its name). Due to the exponentially decaying structure of the channel, perturbation of one element of the matrix to invert $\mathbf{A}$ has only an impact on the coefficients the inverse $\mathbf{A}^ {-1}$ located around the element considered. This comes from the exponentially decaying structure and has for consequence that only local knowledge of the matrix $\mathbf{A}$ is needed to obtain an arbitrary good approximation. 

This property can be put in relation with another useful Theorem which states that an exponetially decaying matrix can be approximated arbiratry well by band limited matrices\cite{Bickel2010}.

We will now show how it is possible to use this property of \emph{local inversibility} to derive a CSI alloction which is more \emph{stingy} in how the CSI is allocated. Thus, the CSI allocation is reduced very strongly so as to fulfill only the most necessary needs.

\subsection{CSI Allocation in Exponentially Decaying Channels}
We consider now the exponentially decaying channel matrix defined in \eqref{eq:ExpDec_1}. We discuss the MG achieved in that case when the $\mu$ tends to zero as $P$ tends to infinity. As already explained for the Wyner model, this is only an artefact to model the pathloss attenuation and at finite SNR, both are finite and we will see by simulations that the approach is efficient in realistic conditions.

Before stating the main result, we start by applying Proposition~\ref{proposition_DecayRate} to the channel matrix $\mathbf{H}$ which is observed to decay with the rate $\gamma=-\log(\mu)$. Thus, we conclude from the proposition that the channel inverse $\mathbf{H}^{-1}$ is exponentially decaying with the decay rate $\gamma'$. By similarity to the structure of the channel matrix, we define $\mu'\triangleq \exp(-\gamma')\in(0,\mu]$.

\begin{theorem}
Considering that $\mu$ tends to zero as $P$ tends to infinity, a sufficient CSI allocation to achieve the maximal MG using RVQ is to quantize the channel vectors $\bm{h}_i^{\He}$ at TX~$j$ with a number of bits equal to:
\begin{equation}
B_i^{(j)}=(K-1)\max\left(\log_2\left(\frac{P}{\mu'^2}\left(\frac{\mu'^{2|i-j|}}{\mu^{2|i-j|}}\right)\right),0\right).
\label{eq:ExpDec_5}
\end{equation} 
 If a uniform scalar quantization is used, and we define the number of bits quantizing $H_{ik}^{(j)}$ by $B_{ik}^{(j)}$, then it is sufficient to achieve the maximal MG to let
\begin{equation}
B_{ik}^{(j)}=\max(\log_2\left(\left(\frac{P}{\mu'^2}\left(\frac{\mu'^{2|i-j|}}{\mu^{2|i-j|}}\right)\mu^{2|k-j|}\right),0\right).
\label{eq:ExpDec_6}
\end{equation}
\label{theorem_CSI_alloc_exp}
\end{theorem} 
\begin{proof}
We first notice that the channel matrix $\mathbf{H}$ defined in \eqref{eq:ExpDec_1} is exponentially decaying with decay constant $\gamma=-\log(\mu)$. The proof will be based on the use of the ``Windows Lemma''~\ref{proposition_window} with the two matrices $\mathbf{H}$ and $\mathbf{H}^{(j)}$. We first notice that it exists a natural number $k'$ such that $\mu'^{k'}<1/P$ with $\mu'\triangleq \exp(-\gamma')$ defined from $\gamma'$ the decay constant for the channel inverse $\mathbf{H}^{-1}$. We define $K_0'$ as the smallest natural number verifying this property which means that
\begin{equation} 
K_0'=\left\lceil \frac{-\log(P)}{\log(\mu)}\right\rceil.
\label{eq:ExpDec_7}
\end{equation} 
Thus, we $K'_0$ fulfills the relation
\begin{equation} 
K_0'-1\leq \frac{-\log(P)}{\log(\mu)}\leq K_0' 
\label{eq:ExpDec_8}
\end{equation} 
which we can rewrite as
\begin{equation} 
\mu'^{K_0'-1}\geq \frac{1}{P}\geq \mu'^{K_0'}. 
\label{eq:ExpDec_9}	
\end{equation} 

We also define the set $\Omega$ as 
{\small\begin{equation}
\Omega\triangleq\{(j+K'_0+\ell,1),\ldots,(j+K'_0+\ell,K),(j-K'_0-\ell,1),\ldots,(j-K'_0-\ell|\forall \ell\in \mathbb{N}\}.
\label{eq:ExpDec_10}
\end{equation}}
In fact, the set $\Omega$ is simply made of all the rows such that there are at least $K'_0-1$ rows between them and the $j$-th row. 

Let consider now that the CSI is allocated so as to fulfill the condition on the channels in Proposition~\ref{proposition_window} with this choice of $\Omega$ and the matrices $\mathbf{H}^{(j)}$ and $\mathbf{H}$. The difference between the $j$-th row of the inverse of $\mathbf{H}^{(j)}$ and the inverse$\mathbf{H}$ is known at TX~$j$ with an error of the order $O(\mu'^{K_0'})$.  Using \eqref{eq:ExpDec} we can further bound it to get
\begin{equation} 
\mu'^{K_0'}\leq \frac{1}{P}. 
\label{eq:ExpDec_11}	
\end{equation} 
In fact, only the $j$-th row of the precoder is used at TX~$j$ such that after a stream by stream normalization, the condition given in \eqref{eq:Wyner_5} is fulfilled and the maximum MG is achieved.

Therefore, we need to prove that the conditions for the channel matrices is fullfilled to conclude the proof. Let consider in a first step the requirement for the CSI on the channel vector $\bm{h}_j$ at TX~$j$. To fullfill the condition of the theorem, the error made on the estimation of the channel vector $\bm{h}_j$ should not be larger than a $O(\mu^{2K'_0})$. Using \eqref{eq:ExpDec_9}, we can then lower bound as
\begin{equation} 
\mu^{2K'_0}=\frac{\mu^{2K'_0}}{\mu'^{2K'_0}}\mu'^{2K'_0}\geq\frac{\mu^{2K'_0}}{\mu'^{2K'_0}}\mu'\frac{1}{P}. 
\label{eq:ExpDec_12}
\end{equation}
To have the estimation smaller than the RHS, it is necesssary to have a number of bits equal to 
\begin{equation} 
B=\left\lceil(K-1)\log_2\left(\frac{\mu^{2K'_0}}{\mu'^{2K'_0}}\mu'\frac{1}{P}\right)\right\rceil 
\label{eq:ExpDec_13}
\end{equation}
when RVQ is used to encode $\bm{h}_j$. The expression in \eqref{eq:ExpDec_13} is exactly the formula given in the Theorem, and the expressions for the other channel vectors $\bm{h}_{\ell}$ with $\ell\neq j$ follow exactly on the same way, only with $K'_0$ replaced by $K'_0-|\ell-j|$. This concludes the proof of \eqref{eq:ExpDec_5}.

The start of the proof of \eqref{eq:ExpDec_6} is exactly the same. The only difference is that to obtain an estimation error of the order $O(\mu^{2K'_0}\mu'/(\mu'^{2K'_0}P)$, less bits are necessart when using scalar quantization with the knowledge of the variance of the coefficients. Let consider without loss of generality the quantization of $\bm{h}_j$. Then, $H_{ji}$ has a variance equal to $\mu^{2|i-j|}$. Thus, only $\log_2(\mu^{2|i-j|})$ less bits are needed to obtain the same accuracy as for the direct element $H_{jj}$. Using this result in \eqref{eq:ExpDec_5} gives \eqref{eq:ExpDec_6}.
\end{proof}
The CSI allocation \eqref{eq:ExpDec_6} is based on the same property of the precoder as in the Wyner model that the off-diagonal elements in the inverse are exponentially decaying. The modified CSI allocation \eqref{eq:ExpDec_6} makes use of the knowledge of the variance of the coefficients inside the channel vectors so that a scalar with a smaller variance is quantized with less bits.

The constant $\mu'$ is not known beforehand but can realistically be estimated based on long term information. Furthermore, we have the strong belief that it is in fact equal to $\mu$. The CSI allocation derived requires much less bits than the conventional one for the BC. Particularly, it also consists in limiting the cooperation to a \emph{local scale} which leads to the very good scaling of the total number of bits when the number of cooperating TXs increases.
\begin{proposition}
Allocating the CSI according to~\eqref{eq:ExpDec_5} leads to a quadratic scaling of the total number of feedback bits in $K$, while using the CSI allocation \eqref{eq:ExpDec_6} leads to a linear scaling, which means that the number of CSI bits by TX does not increases with $K$.
\label{prop_CSI_alloc_Wyner}
\end{proposition} 
\begin{proof} 
The proof of the quadratic scaling of \eqref{eq:ExpDec_5} is exactly the same as for the Wyner model, and in fact is already proven in the proof of~\ref{theorem_CSI_alloc_exp} since only the CSI of $2K'_0$ users is shared with TX~$j$. The proof of the linear scaling of \eqref{eq:ExpDec_6} follows the same pattern as uses the fact that only a finite number of the channel vector~$\bm{h}_i$ are quantized with a positive number of bits.
%Defining $\mu_m= \min(\mu,\mu')$, there exists a $K^m_0$ such that $\mu_m^{K^m_0}<1/P$. The total number of bits $B^{(j)}$ at TX~$j$ is then equal to
%\begin{equation}
%\begin{aligned}
%B^{(j)}
%&=\sum_{i=1}^K\sum_{k=1}^K B_{ik}^{(j)}\\
%&=\sum_{i=1}^K\sum_{k=1}^K \max(\log_2(P\mu^{2|k-j|}\mu_i^{2|i-j|}),0)\\
%&=\sum_{i'=1}^K/2\sum_{k'=1}^K/2 4\max(\log_2(P\mu^{2k'}\mu_i^{2i'}),0)\\
%&=\sum_{\underset{i'+k'<K^m_0}{1<i',k'<K/2}} 4\max(\log_2(P\mu^{2k'}\mu_i^{2i'}),0).
%\end{aligned}
%\end{equation}
%Hence, the summation has only a finite number of terms and the scaling is linear in $K$.
\end{proof}
The restriction on the cooperation to a local scale is very interesting as it allows for the cooperation of very large number of TX/RXs since there is no cost induced: the cooperation remains on the cooperation scale.

%This result leads to very significant CSI reductions particularly in large networks. However, we need to keep in mind that the objective criterium was set so as at fulfill the criterium on the MG which does not protect from losses in rate offset which can have a significant impact at realistic SNR. This could be improved by tuning a margin factor, e.g. set so as to let the error be only one tenth of the noise variance. Furthermore, the results are valid for $\mu$ arbitrarily small which means that for realistic values, some further losses can arise. We will see that this is particularly true for the exponentially decaying case as the number of terms creating interference is much larger than in the Wyner model. Third, the theoretical result did not give the value of $\mu'$ and in the following we will approximate the value of $\mu'$ by $\mu$ as it is the critical value. %Altogether, the results give the basis for the general approach but some fine parameters remain free to be optimized so as to boost the performance in practical systems.

\subsection{Broadcasted CSI}\label{se:ExpDec:BC}
The broadcasting of the CSI by the RXs instead of transmitting it only to its own TX, has appeared recently as an interesting practical scheme to transmit the CSI to the different cooperating TXs as a reduced cost. A fundamental question has been to determine whether the Broadcasting of the CSI leads somehow to an optimal CSI allocation optimized by "nature". Indeed, the intuition is that a TX located far away receives an estimate with low accuracy but also interfers little so that it does not require a very good accuracy. From the analysis of the Wyner model, we know that this is not alway the case. Indeed, a given channel estimate should also be shared to TXs with which this user has no direct link. This is not anymore the case if you remove the intermediate TX/RX pairs. Indeed, the two TXs have an impact on each other because of the interference emitted to intermediate TXs.

In the particular case considered, the "indirect links" are smaller than the direct links such that this leads to the property that broadcasting the CSI leads to achieve close to the maximal MG as giving in Theorem~\ref{theorem_CSI_alloc_Wyner}. Indeed, it can be easily seen by inspection that the rate whish can be obtained through the interfering links correspond to the accuracy required by Theorem~\ref{theorem_CSI_alloc_Wyner}. This is not further investigated for the moment, but we will show that this is confirmed by the simulations.

\subsection{Simulations for the Exponentially Decaying Channels}\label{se:ExpDec:Sim}

In this section, we provide simulations for an exponentially decaying channel defined according to \eqref{eq:ExpDec_1}. We consider a network made of $K=25$ cooperating TXs and we consider the same CSI allocation solutions as for the simulations in the Wyner model in Section~\ref{se:ExpDec:Sim}. For the \emph{Decaying CSI allocation}, we use the expression~\ref{eq:ExpDec_6}, and we also show an additional curve whish correponds to the broadcast of the CSI. Finally, we use $\mu'=\mu$ as an approximation of the true value of $\mu'$.

In Fig.~\ref{R_Linear_25TX_Musqrt02_VarPow}, we show the average rate per user achieved as a function of the SNR. We can observe how the conventional clustering perform badly, as does the uniform power allocation which would become efficient only at extremely high SNR. We can see how the decaying CSI allocation leads to practically no losse compared to the sum rate, while the broadcasting of the CSI also performs very well. note that due to the cooperation at only a local scale, only less than a percent of the total CSI is required by the decaying CSI allocation.

\begin{figure}%[htp!] 
\centering
\includegraphics[width=1\columnwidth]{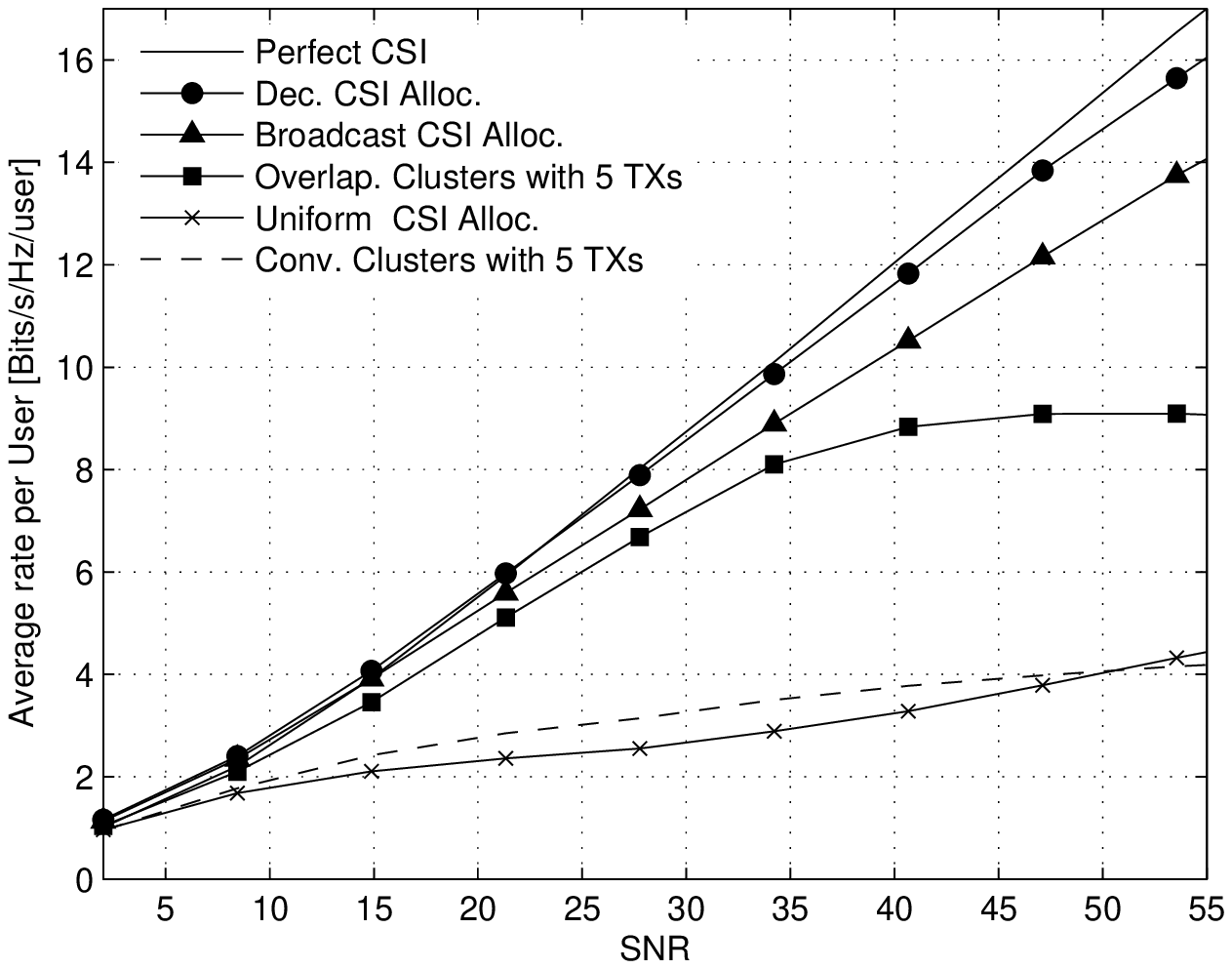}
\caption{Average rate per user in terms of the inter-cell attenuation factor $\mu$ for a normalized transmit power per TX equal to $20$ dB.}
\label{R_Linear_25TX_Musqrt02_VarPow}
\end{figure}
In Fig.~\ref{R_Linear_25TX_P30dB_VarMu}, the sum rate is shown in terms of the intercell attenuation coefficient. We can observe how the decaying CSI allocation performs as the ZF with perfect CSI while it is outperformed by the overlapping cluster allocation as the value of the coefficient increases.
\begin{figure}%[htp!] 
\centering
\includegraphics[width=1\columnwidth]{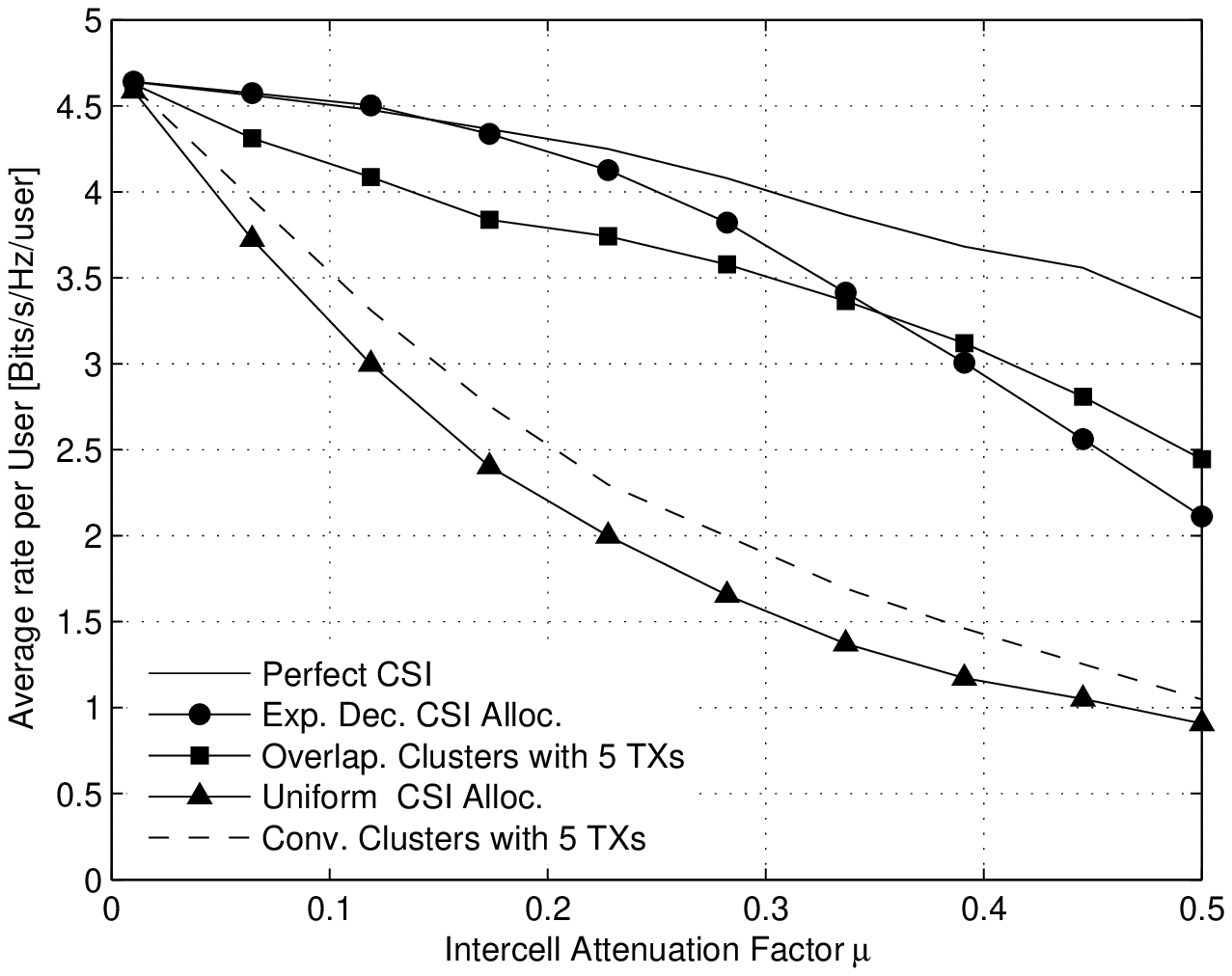}
\caption{Average rate per user in terms of the inter-cell attenuation factor $\mu$ for a normalized transmit power per TX equal to $20$ dB.}
\label{R_Linear_25TX_P30dB_VarMu}
\end{figure}
Finally, in Fig.~\ref{Amp_Linear_25TX_Musqrt02}, we show the logarithm of the amplitude of the precoding coefficient for a given stream normalized by the power used. We also show the logarithm of the amplitude of the channel coefficients with the indices in abscisse ordered so as to correspond to decreasing channel coefficients. We observe that the precoder coefficient are then also decreasing and that the decreasing is also exponential.
\begin{figure}%[htp!] 
\centering
\includegraphics[width=1\columnwidth]{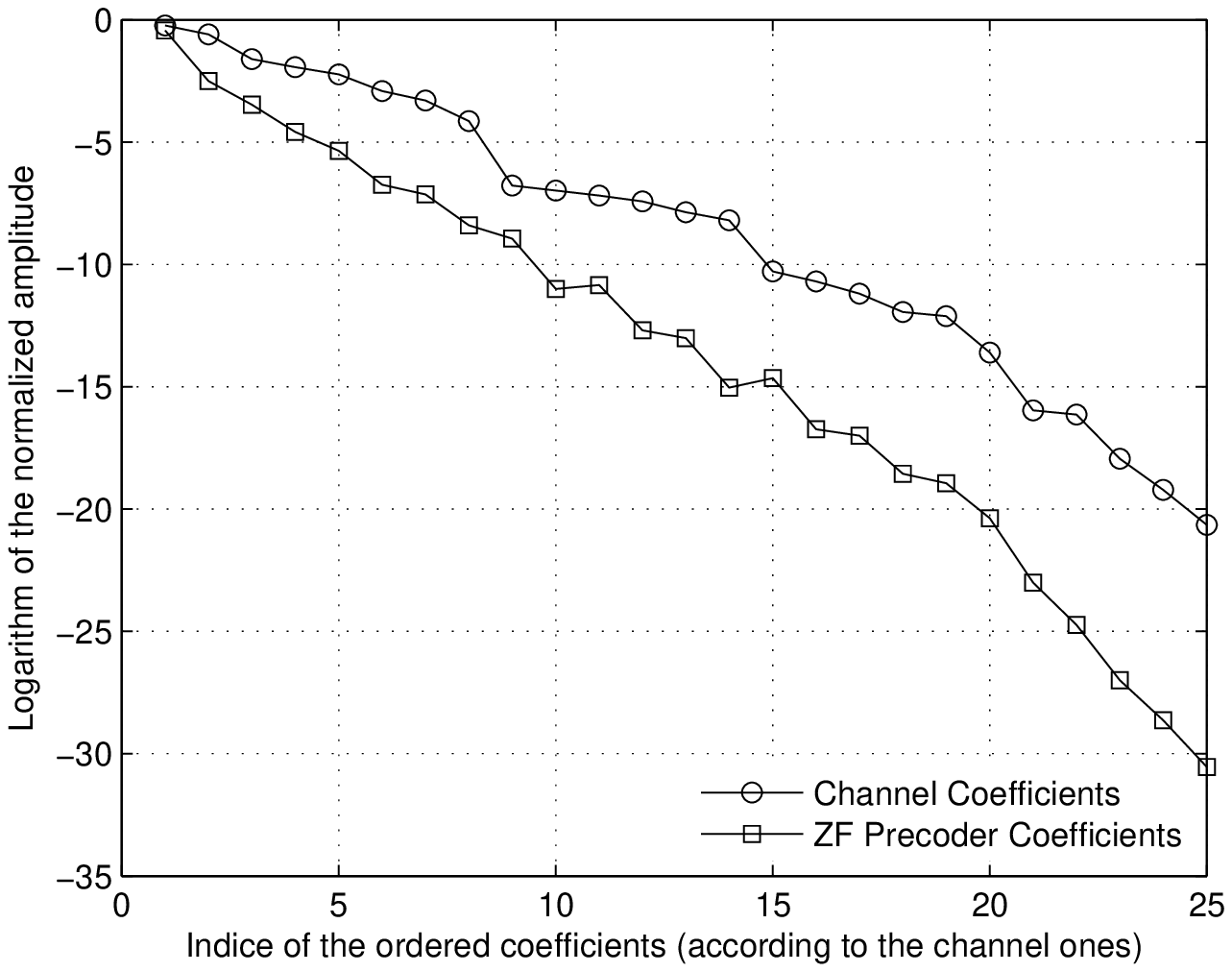}
\caption{Average rate per user in terms of the inter-cell attenuation factor $\mu$ for a normalized transmit power per TX equal to $20$ dB.}
\label{Amp_Linear_25TX_Musqrt02}
\end{figure}

%\section{CSI Allocation in Two Dimensional Cellular Networks}\label{se:Cellular}
%
% 
%\subsection{Simulations for the Cellular Networks}\label{se:Cellular:Simulations}
%
%\begin{figure}%[htp!] 
%\centering
%\includegraphics[width=1\columnwidth]{R_BCCSI_FullR_VarPow.eps}
%\caption{Average rate per user in terms of the inter-cell attenuation factor $\mu$ for a normalized transmit power per TX equal to $20$ dB.}
%\label{R_BCCSI_FullR_VarPow}
%\end{figure}
%
%\begin{figure}%[htp!] 
%\centering
%\includegraphics[width=1\columnwidth]{Amp_BCCSI_FullR.eps}
%\caption{Average rate per user in terms of the inter-cell attenuation factor $\mu$ for a normalized transmit power per TX equal to $20$ dB.}
%\label{Amp_BCCSI_FullR}
%\end{figure}
%
%
%\begin{figure}%[htp!] 
%\centering
%\includegraphics[width=1\columnwidth]{R_BCCSI_HalfR_VarPow.eps}
%\caption{Average rate per user in terms of the inter-cell attenuation factor $\mu$ for a normalized transmit power per TX equal to $20$ dB.}
%\label{R_BCCSI_HalfR_VarPow}
%\end{figure}
%
%\begin{figure}%[htp!] 
%\centering
%\includegraphics[width=1\columnwidth]{Amp_BCCSI_HalfR.eps}
%\caption{Average rate per user in terms of the inter-cell attenuation factor $\mu$ for a normalized transmit power per TX equal to $20$ dB.}
%\label{Amp_BCCSI_HalfR}
%\end{figure}

\section{Conclusion}
In this work, we have studied the optimization of the CSI allocation in the Distributed CSI-MIMO channel. For the Wyner model, we have derived a CSI allocation achieving good performance with a large reduction in the amount of CSI feedback needed. This is especially true for large networks as the scaling of the total number of CSI feedback bits is reduced so that the number of bits allocated to a TX does not increase with the number of cooperating TXs. This approach has then been extended to the conventional Rayleigh fading setting if the channel is exponentially decaying. By simulations, we have confirmed the strong potential of the proposed approach where the CSI allocation is strongly reduced (up to only a few percents of the original CSI allocation are needed) at the cost of reduced performance losses. Studying further the application of this scheme to more practical transmission scheme is a very intereting direction of research. Particularly, we have shown that broadcast the CSI leeds to a achieve close to the CSI allocation derived.
\bibliographystyle{IEEEtran}
\bibliography{Literatur}
\end{document}